\documentclass[a4paper,UKenglish,cleveref, autoref, thm-restate]{lipics-v2021}
\hideLIPIcs  %uncomment to remove references to LIPIcs series (logo, DOI, ...), e.g. when preparing a pre-final version to be uploaded to arXiv or another public repository

\usepackage{amsmath}
\usepackage{amssymb}
\usepackage{amsthm}
\usepackage{braket}
\usepackage{xspace}
\usepackage[ruled,vlined,linesnumbered]{algorithm2e}
\usepackage{cite}
\usepackage{soul}
\usepackage{color}

\newcommand\smlg{\textsf{SMLG}\xspace}
\newcommand\ov{\textsf{OV}\xspace}
\newcommand\ovh{\textsf{OVH}\xspace}

\newcommand\seth{\textsf{SETH}\xspace}

\newtheorem{problem}{\textbf{Problem}}

\title{From Bit-Parallelism to Quantum String Matching for Labelled Graphs} %TODO Please add

\author{Massimo Equi}{Department of Computer Science, University of Helsinki, Finland}{massimo.equi@helsinki.fi}{https://orcid.org/0000-0001-8609-0040}{}

\author{Arianne Meijer - van de Griend}{Department of Computer Science, University of Helsinki, Finland}{arianne.vandegriend@helsinki.fi}{https://orcid.org/0000-0001-5946-0958}{}

\author{Veli M\"akinen}{Department of Computer Science, University of Helsinki, Finland}{veli.makinen@helsinki.fi}{https://orcid.org/0000-0003-4454-1493}{}

\authorrunning{M. Equi, A. Meijer - van de Griend, V. M\"akinen} %TODO mandatory. First: Use abbreviated first/middle names. Second (only in severe cases): Use first author plus 'et al.'

\Copyright{Massimo Equi, Arianne Meijer - van de Griend and Veli M\"akinen} %TODO mandatory, please use full first names. LIPIcs license is "CC-BY";  http://creativecommons.org/licenses/by/3.0/


 %TODO mandatory; please add comma-separated list of keywords

\category{} %optional, e.g. invited paper

\nolinenumbers %uncomment to disable line numbering

\EventEditors{Laurent Bulteau and Zsuzsanna Lipt\'{a}k}
\EventNoEds{2}
\EventLongTitle{34th Annual Symposium on Combinatorial Pattern Matching (CPM 2023)}
\EventShortTitle{CPM 2023}
\EventAcronym{CPM}
\EventYear{2023}
\EventDate{June 26--28, 2023}
\EventLocation{Marne-la-Vall\'{e}e, France}
\EventLogo{}
\SeriesVolume{259}
\ArticleNo{10}
\begin{document}

\maketitle

\begin{abstract}
Many problems that can be solved in quadratic time have bit-parallel speed-ups with factor $w$, where $w$ is the computer word size. A classic example is computing the edit distance of two strings of length $n$, which can be solved in $O(n^2/w)$ time. In a reasonable classical model of computation, one can assume $w=\Theta(\log n)$, and obtaining significantly better speed-ups is unlikely in the light of conditional lower bounds obtained for such problems.

In this paper, we study the connection of bit-parallelism to quantum computation, aiming to see if a bit-parallel algorithm could be converted to a quantum algorithm with better than logarithmic speed-up. We focus on \emph{string matching in labeled graphs}, the problem of finding an exact occurrence of a string as the label of a path in a graph. This problem admits a quadratic conditional lower bound under a very restricted class of graphs (Equi et al. ICALP 2019), stating that no algorithm in the classical model of computation can solve the problem in time $O(|P||E|^{1-\epsilon})$ or $O(|P|^{1-\epsilon}|E|)$. We show that a simple bit-parallel algorithm on such restricted family of graphs (level DAGs) can indeed be converted into a realistic quantum algorithm that attains subquadratic time complexity $O(|E|\sqrt{|P|})$.
\end{abstract}

\section{Introduction}
Exact string matching problem is to decide if a pattern string $P$ appears as a substring of a text string $T$. In the classical models of computation, this problem can be solved in $O(|P|+|T|)$ time \cite{KMP77}. Different quantum algorithms for this basic problem have been developed \cite{NY21,RV03,SM21}, resulting into different solutions, the best of which finds a match in $O(\sqrt{|T|}(\log^2|T|+\log|P|))$ time~\cite{NY21} with high probability. These assume the pattern and text are stored in quantum registers, requiring thus $O(|P|+|T|)$ qubits to function. Moreover, these approaches may rely on applying a linear number of quantum gates in parallel on different qubits. For example, Niroula and Nam \cite{NY21} perform $O(\log(|T|))$ rounds of parallel swaps, executing $O(|T|)$ swaps in parallel per round.

In the classical models of computation, an analogy for these assumptions is to assume that the text has been preprocessed for subsequent queries. For example, one can build a Burrows-Wheeler transform -based index structure for the text in time $O(|T|)$ \cite{BCKM20}, assuming $T \in \{1,2,\ldots, \sigma\}^*$, where $\sigma\leq |T|$. Then, one can query the pattern from the index in $O(|P|\log \log \sigma)$ time \cite[Theorem 6.2]{BCKM20}. In this light, quantum models can offer only limited benefit over the classical models for exact string matching.

Motivated by this difficulty in improving linear-time solvable problems using quantum approaches, let us consider problems known to be solved in quadratic time. For example, approximate string matching problem is such a problem: decide if a pattern string $P$ is within edit distance $k$ from a substring of a text string $T$, where edit distance is the number of single symbol insertions, deletions, and substitution needed to convert a string to another. This problem can be solved using \emph{bit-parallelism} in $O(\lceil |P|/w \rceil |T|)$ time \cite{Mye99}, under the Random Access Memory (RAM) model with computer word size $w$. A reasonable assumption is that $w=\Theta(\log |T|)$, so that this model reflects the capacity of classical computers. Thus, when $|P|=|T|=n$, this bit-parallel algorithm for approximate string matching takes time at least $\Omega(n^{2-\epsilon})$ for all $\epsilon>0$, as $\log n=o(n^\epsilon)$ for all $\epsilon>0$. It is believed that this quadratic bound cannot be significantly improved, as there is a matching conditional lower bound saying that if approximate pattern matching could be solved in time $O(n^{2-\epsilon})$ with some $\epsilon>0$, then the Orthogonal Vector Hypothesis (\ovh) and thus the Strong Exponential Time Hypothesis (\seth) would not hold \cite{BI18}. As these hypotheses are about classical models of computation, it is natural to ask if the quadratic barrier could be broken with quantum computation.

In this quest for breaking the quadratic barrier, we study another problem with a bit-parallel solution and a conditional lower bound. Consider exact pattern matching on a graph, that is, consider deciding if a pattern string $P\in \Sigma^*$ equals a labeled path in a graph $G=(V,E)$, where $V$ is the set of nodes and $E$ is the set of edges. Here we assume the nodes $v$ of the graph are labeled by $\ell(v)\in \Sigma$ and a path $v_1 \to v_2 \to \cdots v_{t}$, $(v_i,v_{i+1}) \in E$ for $1\leq i<t$, spells string $\ell(v_1)\ell(v_2)\cdots \ell(v_t)$. There is an \ovh lower bound conditionally refuting an $O(|P||E|^{1-\epsilon}$) or $O(|P|^{1-\epsilon}|E|$) time solution \cite{EGMT19}. This conditional lower bound holds even if graph $G$ is a \emph{level} DAG: for every two nodes $u$ and $v$, holds the property that every path from $u$ to $v$ has the same length. On DAGs, this string matching on labeled graphs (\smlg) problem can be solved in $O(\lceil |P|/w \rceil |E|)$ time \cite{RMM19} in the bit-parallel model, so the status of this problem is identical to that of approximate pattern matching on strings. However, the simplicity of the bit-parallel solution for \smlg on level DAGs enables a connection to quantum computation. We consider a specific model of quantum computation, the Quantum Random Access Memory (QRAM) model \cite{GLM2008}, in which we have access to ``quantum arrays'', and we assume that integer values like $|P|$, $|V|$ or $|E|$ fit into a (quantum) memory word. Under this model, we turn the bit-parallel solution into a quantum algorithm that solves \smlg on level DAGs with high probability in $O(|E|\sqrt{|P|})$ time, breaking through the classical quadratic conditional lower bound.

Classical conditional lower bounds are not new to be broken by quantum computing. For example, the quadratic Orthogonal Vectors problem itself can be solved in subquadratic time (linear using QRAM) using quantum computing. This is not the only problem to have a better-than-quadratic solution in the quantum realm \cite{V2019}. Nevertheless, to the best of our knowledge, we are the first to propose a subquadratic time algorithm for \smlg, even if restricted to a specific class of graphs. Moreover, the translation of a bit-parallel strategy to a quantum-parallel one is an original technique, and we are not aware of any other work utilising it.

An earlier work \cite{DGT22} provided a quantum algorithm solving \smlg in time $O(\sqrt{|V||E|}|P|)$. When the graph is non-sparse, that is $|V|=O(\sqrt{|E|})$, the time complexity becomes $O(|E|^\frac34|P|)$, which is an improvement over classical algorithms. We offer a different kind of trade-off, limiting ourselves to a special class of graphs, but obtaining a better time complexity. We also note that, even if no subquadratic classical algorithm exists for non-sparse graphs, the existing classical reduction from \ov \cite{EGMT19} produces a sparse level DAG, for which our quantum algorithm runs in subquadratic time.

As mentioned above, in some previous works \cite{NY21,SM21,RV03}(and references in \cite{SM21}) algorithms have been proposed to solve string matching in plain text in the QRAM model, under the assumption that a large number of quantum gates, possibly linear, can be applied in parallel when acting on different qubits. We find this assumption to be too restrictive, as even the classical RAM model does not adopt it, since in such a model of computation many operations would become trivial. Instead, our algorithm works without the need for such an assumption.

The paper is structured as follows. We revisit exact pattern matching and derive a simple quantum algorithm for it, in order to introduce the quantum machinery. Then we give a brute-force quantum algorithm for \smlg, which we later improve on level DAGs. This improvement is based on extending the Shift-And algorithm \cite{B-YG1992}, whose quantum version we extend for level DAGs. 

In what follows, we assume the reader is familiar with the basic notions in quantum computing as covered in textbooks \cite{NC10QCQI}.

\section{Preliminaries}

An \emph{alphabet} $\Sigma$ is a set of \emph{characters}. Throughout the paper we assume $\Sigma$ is ordered, i.e., for each $a,b \in \Sigma$ we can decide if $a<b$. A sequence $P\in \Sigma^n$ is called a \emph{string} and its length is denoted $n=|P|$. We denote integers $i,i+1,\ldots,j$ as interval $[i..j]$ and represent a string $P$ as an array $P[0..n-1]$, where $P[i]\in \Sigma$ for $0\leq i\leq n-1$, as in this work all indexes start from $0$. String $P[i..j]$ is called a \emph{substring} and string $P[0..i]$ a \emph{prefix} of $P$. With bit-vectors discussed next, we use $0$-based indexing.

Let $B$ be a $w$-bit integer interpreted as string $B[0..w-1]$ from alphabet $\{0,1\}$ such that $B=\sum_{i=0}^{w-1} B[i]\cdot 2^i$. We call $B$ a \emph{bit-vector}. Given two bit-vectors $B$ and $C$, we define the following Boolean operations $A=B\land C$, $O=B\lor C$, and $N=\neg B$ as follows: $A[i]=1$ iff $B[i]=C[i]=1$, $O[i]=1$ iff
$B[i]=1$ or $C[i]=1$, and $N[i]=1$ iff $B[i]=0$. When bit-vector content is visualized, we list the most significant bit first, i.e., $B[w-1] B[w-2] \cdots B[0]$. With this in mind, we define the left-shifts $L=B \ll k$ and right-shifts $R=B \gg k$ as follows: $L[i+k]=B[i]$ and $R[i]=B[i+k]$. Here values out of the domain of the bit-vectors are assumed to be $0$. Logarithms are assumed to be in base two: $\log n=\log_2 n$.

In \emph{directed labelled graph} (DAG) $G=(V,E,\ell)$, $V$ is the set of nodes, $E$ is the sets of vertices, and $\ell:V\rightarrow\Sigma$ is a labelling function that assigns a character of the alphabet to each node. We assume the nodes to be indexed as $v_0, v_1, \ldots v_{n-1}$ in topological order, where $n=|V|$. For $v_i \in V$, $\ell(v_i)$ is its label. Set of nodes $in(v_i) = \{j \,|\, (v_j,v_i)\in E\}$ contains the indexes of the in-neighbours of $v_i$, and $D_i=|in(v_i)|$ is the in-degree of $v_i$. If, for $0\leq d \leq D_i-1$, $v_k$ is the $d$-th in-neighbour of $v_i$ according to the topological indexing that we defined above, we express this fact using notation $k=in_i(d)$, where $in_i:[0,D_i-1]\rightarrow[0,n-1]$.

In this work, we study the problem of \emph{string matching in labelled graphs}, that consists in finding a match for a pattern string $P[0..m-1]$ in a labelled graphs $G$ over alphabet $\Sigma$, where $P$ has a \emph{match} in $G$ if there is a path $v_1, \ldots, v_k$ such that $P = \ell(v_1) \cdots \ell(v_k)$ (we also say that $P$ \emph{occurs} in $G$, and that $v_1, \ldots, v_k$ is an \emph{occurrence} of $P$). Notice that if $|P|=1$, a classic visit of the graph solves the problem in linear time, thus we always assume $|P| \geq 2$.

\begin{problem}[String Matching in Labeled Graphs (\smlg)]
\item{\textsc{input}:} A labeled graph $G = (V,E,L)$ and a pattern string $P$, both over an alphabet $\Sigma$.
\item{\textsc{output}:} \emph{True} if and only if there is at least one occurrence of $P$ in $G$.
\end{problem}

\section{Quantum Notation and Preliminaries}
In quantum computing, data is represented in quantum bits (qubits), the quantum analogue to classical bits. 
A qubit can be in two states, denoted as $\ket{0} =$ ($\begin{smallmatrix}1\\0\end{smallmatrix}$) and $\ket{1} =$ ($\begin{smallmatrix}0\\1\end{smallmatrix}$) but, unlike a classical bit, it can also be a linear combination of the two states, a \textit{superposition}: $\ket{\psi} = \alpha\ket{0} + \beta\ket{1}$.
The complex values $\alpha$ and $\beta$ are called the \textit{amplitudes} of $\ket{\psi}$. Measuring a qubit in superposition will result in either $\ket{0}$ or $\ket{1}$ with probabilities $|\alpha|^2$ and $|\beta|^2$, respectively. Note that this notation can easily be generalised to integer states $\ket{n}$ using the tensor product between the quantum states of the binary representation on $n$: $\ket{n} = \bigotimes_{i \in binary(n)}\ket{i}$, and in this case we use the term \emph{quantum register}. Throughout the paper, we will use notation $\ket{q}_Q$ to denote that qubit $Q$ is in state $\ket{q}$. We use lower case letters for quantum states and capital letters for qubits.

In this work, we mainly use the NOT gate $X$, the controlled NOT $CX$, and the Toffoli gate $CCX$. We also apply an OR gate, that computes a logical \emph{or} between two qubits and stores the results in a third quibit. This can easily be obtained with a simple combination of $X$ gates with a Toffoli gate.

Furthermore, to define some quantum states, we use Kronecker's delta function $\delta_{x,y}$, which is $\delta_{x,y}=1$ if $x=y$ and $\delta_{x,y}=0$ otherwise. Given superposition $\ket{\psi}=\sum_{i=0}^{n-1}\alpha_i\ket{i}_I\ket{\delta_{c,i}}_Q$, the delta function specifies that qubit $Q$ is in state $\ket{1}$ iff $i=c$, as in the following example
\[
\sum_{i=0}^{3}\alpha_i\ket{i}_I\ket{\delta_{0,i}}_Q = \alpha_0\ket{0}_I\ket{1}_Q+\alpha_1\ket{1}_I\ket{0}_Q+\alpha_2\ket{2}_I\ket{0}_Q+\alpha_3\ket{3}_I\ket{0}_Q
\]
where $I$ is a quantum register of at least two qubits.

We assume to have a \emph{quantum random access memory} (QRAM) able to use a quantum register as an index to access classical data. Let $m_0, m_1, \ldots, m_{n-1}$ be the data stored in QRAM $M$. Given quantum register $I$, the operation that reads data from $M$ into quantum register $Q$ initialized to $\ket{0}$ using $I$ as index is defined as follows \cite{GLM2008}:
\[
\sum_{i=0}^{n-1} \alpha_i\ket{i}_I\ket{0}_Q \xrightarrow{\text{QRAM read}} \sum_{i=1}^n \alpha_i \ket{i}_I\ket{0 \oplus m_i}_Q = \sum_{i=1}^n \alpha_i \ket{i}_I\ket{m_i}_Q.
\]
Notice that this is a unitary operation, and thus reading the same data into the same register twice will reset such a register to the value it had before performing the reading operation. In terms of time complexity, the execution of the read operation is proportional to the number of qubits in quantum register $I$. Under the Word-QRAM model with memory-word size $O(\log n)$ for inputs of size $n$, we can assume to be able to perform a QRAM read operation in $O(1)$, because $O(\log n)$ qubits are enough for register $I$ to index an input of size $n$. Indeed, this reflect the same assumption of the classical Word-RAM model, where operations on memory words are assumed to be constant.

\section{String Matching in Plain Text}
A quantum computer, with access to QRAM, can solve the problem of finding an exact match for a pattern string $P$ into a text string $T$ in time $O(|P|\sqrt{|T|})$, with high probability. We explain a simple solution to this problem. Let $|T|=n$ and $|P|=m$, then $T=t_0 t_1\cdots t_{n-1}$ and $P=p_0 p_1\cdots p_{m-1}$ are two strings defined over a binary alphabet, that is $t_i,p_j \in \{0,1\}$ for $0\leq i\leq n-1$ and $0\leq j\leq m-1$. We use qubits $C_T$ and $C_P$ initialized to $\ket{0}$ to track the current characters of $T$ and $P$, and we assume to have the text and the pattern stored in qubits in the following way:
\[
\ket{0}_{C_T}\ket{0}_{C_P}
\ket{t_0}_{T_0} \ket{t_1}_{T_1}\cdots \ket{t_{n-1}}_{T_{n-1}}\ket{p_0}_{P_0} \ket{p_1}_{P_1}\cdots \ket{p_{m-1}}_{P_{m-1}}.
\]
We also use auxiliary qubits $A_{-1},A_0, A_1 \cdots A_{m-1}$, and quantum registers $I$, $J$, and $Q$, all three of $\log n$ qubits. We initialize $A_{-1}$ and $Q$ to $\ket{1}$, while $A_0, A_1 \cdots A_{m-1}$, $I$ and $J$ are all initialized to $\ket{0}$. We prepare quantum register $I$ in an equally balanced superposition spanning all the text positions, that is $\ket{0}_I \rightarrow 1/\sqrt{n}\sum_{i=0}^{n-1}\ket{i}_I$, assuming $n$ to be a power of $2$, without loss of generality. If this is not the case, we generate a superposition as large as the first power of two greater than $n$, then standard techniques can be applied to handle the additional substates, as explained in Appendix~\ref{appendix:power-of-two}.

Each individual state $\ket{i}$ in the superposition represents a computation starting at position $i$ in the text. In each of these computations, we scan $T[i..i+m-1]$ and try to match each character with $P[0..m-1]$, storing the intermediate results of such comparisons in registers $A_0, A_1, \cdots, A_{m-1}$. More precisely, at iteration $j$, $0 \leq j \leq m-1$, we compute a logical \emph{xor} between $t_{i+j}$ and $p_j$ storing the result in $C_P$ via a $CX$ gate with control $C_T$ and target $C_P$. Then, we apply a $X$ gate to $C_P$, which now stores $\ket{\lnot(t_{i+j} \oplus p_j)}_{C_P}=\ket{t_{i+j} = p_j}_{C_P}$. At this point, we apply a Toffoli gate with controls $C_P$ and $A_{j-1}$, storing the value in target qubit $A_j$. We now reset $C_T$ and $C_P$ to $\ket{0}$ by applying to them the same gates again, but in reverse order. As last step in iteration $j$, we increase both $I$ and $J$ by $1$ by performing transformation $1/\sqrt{n}\sum_{i=0}^{n-1}\ket{1}_Q\ket{i}_I\ket{j}_J\rightarrow1/\sqrt{n}\sum_{i=0}^{n-1}\ket{1}_Q\ket{i + 1}_I\ket{j + 1}_J$ (this of course requires two separate addition operations), where the addition is intended to be modulo $2^n$. This allows us to read the next character of the pattern at the next iteration.

After the last iteration, we can run Grover's operator \cite{Grover96} where the marked items are represented by $\ket{a_{m-1,i}}_{A_{m-1}}=\ket{1}$, and then measure register $\ket{I}$ to locate the ending position of a match. Of course, we do not know the exact number of marked items, and we address this issue by guessing the number of items and rerunning the whole algorithm a constant number of times. We illustrate the entire procedure in Algorithm~\ref{alg:quantumPlainText}.

The algorithm is correct because, after each iteration of the \emph{for} loop, we correctly keep track of the positions of the text that are active matches for the current prefix of the pattern.
\begin{lemma}
\label{lemma:PlainTextInvariant}
After iteration $j$ of the $for$ loop of Algorithm~\ref{alg:quantumPlainText}, let qubits $I$ and $A_j$ be in superposition $1/\sqrt{n}\sum_{i=0}^{n-1}\ket{i}_I\ket{a_{j,i}}_{A_j}$. Then, $\ket{a_{j,i}}_{A_j}=\ket{1}$ if and only if $T[i..i+j]=P[0..j]$, where  $0\leq j \leq m-1$ and $0\leq i \leq n-1$.
\end{lemma}
\begin{proof}
At iteration $0$, after applying gates $CX$ and $X$, $C_P$ stores $\ket{\lnot(T[i] \oplus P[0])}_{C_P}$ and $A_{-1}$ stores $\ket{1}_{A_{-1}}$, thus the Toffoli gate simply copies value $\lnot(T[i] \oplus P[0])$ to $A_0$. Because we are working with a binary alphabet, $\lnot(T[i] \oplus P[0])$ equals $T[i]=P[0]$, and thus we obtain superposition $1/\sqrt{n}\sum_{i=0}^{n-1}\ket{i}_I\ket{T[i]=P[0]}_{A_0}$.

At iteration $j$, we assume by induction that register $A_{j-1}$ stores $\ket{a_{j-1,i}}_{A_{j-1}}=\ket{1}$ if and only if $T[i \ldots i+j-1]=P[0 \ldots j-1]$. Gates $CX$ and $X$ compute $\lnot(T[i+j] \oplus P[j])$ storing it in $C_P$. We then apply the Toffoli gate with controls $C_P$ and $A_{j-1}$, and target $A_j$, obtaining superposition
$1/\sqrt{n}\sum_{i=0}^{n-1}\ket{i}_I\ket{a_{j-1,i} \land T[i+j]=P[j]}_{A_j}$.
Thus, $\ket{a_{j,i}}_{A_j}=\ket{a_{j-1,i} \land T[i+j]=P[j]}_{A_j}$ is $\ket{1}$ if and only if $T[i..i+j]=P[0..j]$.
\end{proof}

As mentioned above, we have to be careful in running Grover's search algorithm at the end of Algorithm~\ref{alg:quantumPlainText}. We defer these details to the full proof of Theorem~\ref{theorem:QuantumSMLGCorrectness} given in Appendix~\ref{appendix:theorem5-full-proof}. For now, we assume that we are able to retrieve with arbitrarily high probability $1-(7/8)^c$ a marked substate representing a match. Combining this with Lemma~\ref{lemma:PlainTextInvariant}, we obtain the claimed result.
\begin{theorem}
Given a text string $T$, pattern string $P$ and integer $c>0$, Algorithm~\ref{alg:quantumPlainText} finds a match for $P$ in $T$ in time $O(c(|P|\sqrt{|T|}))$. If there is no match, the algorithm returns a negative answer with probability $p=1$. If there is at least one match, the algorithm returns the index of the last position of a match with probability $p>1-(7/8)^c$.
\end{theorem}
\begin{proof}
For the correctness, consider Lemma~\ref{lemma:PlainTextInvariant} where $j=m=|P|$, which is the number of times we run the \emph{for} loop. In this case, $\ket{a_{|P|,i}}_{A_{|P|}}=\ket{1}$ if and only if $T[i..i+|P|-1]=P[0..|P|-1]$. Thus, measuring these substates yields a correct solutions. The details of how to perform such a measurement respecting the time complexity and probability of success are deferred to the full proof of Theorem~\ref{theorem:QuantumSMLGCorrectness} in Appendix~\ref{appendix:theorem5-full-proof}.
\end{proof}

\begin{algorithm}
\caption{An algorithm for solving exact string matching in plain text that, using QRAM, achieves $O(|P|\sqrt{|T|})$ time complexity. The details of how to handle Grover's search at the end are given in Theorem~\ref{theorem:QuantumSMLGCorrectness}, whose full proof is deferred to Appendix~\ref{appendix:theorem5-full-proof}.}
\label{alg:quantumPlainText}
    \KwIn{Text $T$ stored as $\ket{t_0}_{T_0} \ket{t_1}_{T_1}\cdots \ket{t_{n-1}}_{T_{n-1}}$, pattern string $P$ stored as $\ket{p_0}_{P_0} \ket{p_1}_{P_1}\cdots \ket{p_{m-1}}_{P_{m-1}}$, integer $c$}
    \KwOut{A position of $T$ where a match for $P$ ends, if any}
    
    \For{$c$ times}{
    Initialize quantum registers $I, J,A_0, A_1 \cdots A_{m-1}$ as $ \ket{0}_I \ket{0}_J \ket{0}_{A_0}\ket{0}_{A_1}\ket{0}_{A_{m-1}}$\;    
    Initialize quantum register $A_{-1}$ and $Q$ as $\ket{1}_{A_{-1}}$ and $\ket{1}_Q$\;
    \smallskip
    
    \tcp{Apply $H^{\otimes \log n}$ to register $I$}
    $\ket{0}_I \rightarrow \frac{1}{\sqrt{n}}\sum_{i=0}^{n-1} \ket{i}_I$\;
    \smallskip
    
    \For{$m$ times}{\label{alg:start-oracle-text}
        \tcp{Read $T[i]$ in $C_T$ and $P[j]$ in $C_P$ using registers $I$ and $J$ as indexes}
        $\frac{1}{\sqrt{n}}\sum_{i=0}^{n-1}\ket{i}_I\ket{j}_J\ket{0}_{C_T}\ket{0}_{C_P} \rightarrow \frac{1}{\sqrt{n}}\sum_{i=0}^{n-1}\ket{i}_I\ket{j}_J\ket{t_i}_{C_T}\ket{p_j}_{C_P}$\;
        \medskip
        
        \tcp{Apply $CX$ with control $C_T$ and target $C_P$}
        $\frac{1}{\sqrt{n}}\sum_{i=0}^{n-1}\ket{t_i}_{C_T}\ket{p_j}_{C_P} \rightarrow\frac{1}{\sqrt{n}}\sum_{i=0}^{n-1}\ket{t_i}_{C_T}\ket{t_i \oplus p_j}_{C_P}$\;
        \medskip
        
        \tcp{Apply $X$ to $C_P$}
        \mbox{$\frac{1}{\sqrt{n}}\sum_{i=0}^{n-1}\ket{t_i \oplus p_j}_{C_P}\rightarrow\frac{1}{\sqrt{n}}\sum_{i=0}^{n-1}\ket{\lnot (t_i \oplus p_j)}_{C_P}=\frac{1}{\sqrt{n}}\sum_{i=0}^{n-1}\ket{t_i = p_j}_{C_P}$\;}
        \medskip
        
        \tcp{Apply Toffoli with controls $C_P$ and $A_{j-1}$, and target $A_j$}
        $\frac{1}{\sqrt{n}}\sum_{i=0}^{n-1}\ket{t_i = p_j}_{C_P}\ket{a_{j-1}}_{A_{j-1}}\ket{0}_{A_{j}} \rightarrow \frac{1}{\sqrt{n}}\sum_{i=0}^{n-1}\ket{t_i = p_j}_{C_P}\ket{a_{j-1}}_{A_{j-1}}\ket{(t_i = p_j) \land a_{j-1}}_{A_{j}}$\;
        \medskip
        
        \tcp{Reset $C_T$ and $C_P$ to $\ket{0}$ via uncomputation}
        $\frac{1}{\sqrt{n}}\sum_{i=0}^{n-1}\ket{\lnot (t_i \oplus p_j)}_{C_P}\rightarrow\frac{1}{\sqrt{n}}\sum_{i=0}^{n-1}\ket{t_i \oplus p_j}_{C_P}$\;
        
        $\frac{1}{\sqrt{n}}\sum_{i=0}^{n-1}\ket{t_i}_{C_T}\ket{t_i \oplus p_j}_{C_P}\rightarrow\frac{1}{\sqrt{n}}\sum_{i=0}^{n-1}\ket{t_i}_{C_T}\ket{p_j}_{C_P}$\;
        $\frac{1}{\sqrt{n}}\sum_{i=0}^{n-1}\ket{i}_I\ket{j}_J\ket{t_i}_{C_T}\ket{p_j}_{C_P} \rightarrow$
        \mbox{$\frac{1}{\sqrt{n}}\sum_{i=0}^{n-1}\ket{i}_I\ket{j}_J\ket{t_i \oplus t_i}_{C_T}\ket{p_j \oplus p_j}_{C_P} = \frac{1}{\sqrt{n}}\sum_{i=0}^{n-1}\ket{i}_I\ket{j}_J\ket{0}_{C_T}\ket{0}_{C_P}$\;}
        \medskip
        
        \tcp{Increment indexes $I$ and $J$}
        $\frac{1}{\sqrt{n}}\sum_{i=0}^{n-1}\ket{1}_Q\ket{i}_I\ket{j}_J \rightarrow \frac{1}{\sqrt{n}}\sum_{i=0}^{n-1}\ket{1}_Q\ket{i \oplus 1}_I\ket{j + 1}_J$\;
    }
    \smallskip

        Apply gate $Z$ to qubit $R_{n-1}$, so that the sign of the amplitude is flipped if $\ket{r_{n-1,j}}_{R_{n-1}}=\ket{1}$\label{alg:end-oracle-text}\;
        
        \smallskip
        Choose $K\in [0,|P|]$ uniformly at random\;

        \smallskip
        Run Grover's iterate operator the optimal number of times assuming to have $K$ solutions, with the oracle function being lines~\ref{alg:start-oracle-text}-\ref{alg:end-oracle-text} of this algorithm\;
        
        \smallskip
        Measure $R_{n-1}$ into classical register $R_{cl}$\;
        \If{$R_{cl}=1$}{Measure $I$ into classical register $I_{cl}$ and \Return{$I_{cl}$}}
    }
    \Return{$no$}
\end{algorithm}

\begin{figure}
    \centering
    \includegraphics{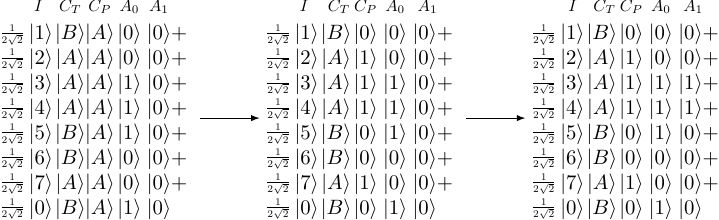}
    \caption{An example of the evolution of the superposition after one iteration of Algorithm~\ref{alg:quantumPlainText}. The first arrow represents the application of a $CX$ gate with control $T_1$ and target $P_1$, and the application of a $X$ gate on $P_1$. The second arrow represents the application of a Toffoli gate with controls $P_1$ and $A_0$, and target $A_1$. Intuitively, in the first step we are checking that $T[i+j]=P[j]$; in the second step we combine the result of this check with the contribution of the previous iteration(s). Characters $A$ and $B$ are to be considered binary values. }
    \label{fig:Alg1IterationStep}
\end{figure}

\section{String Matching in Labeled Graphs}

\subsection{Quantum Brute-force Algorithm for \smlg}
In \smlg we are given pattern string $P$ with characters in alphabet $\Sigma$ and a node-labeled graph $G=(V,E)$, with labelling function $\ell: V \rightarrow \Sigma$. We are asked  to find a path (or, actually, a walk) $\pi = v_1, v_2, \ldots, v_{|P|}$ in $G$ such that $\ell(v_1) \circ \ell(v_2) \circ \ldots \circ \ell(v_{|P|}) = P$, where $\circ$ denotes string concatenation.

One could try to obtain a quantum algorithm for \smlg by generalizing the idea we presented for plain text. The idea would be to list all possible paths of length $|P|$ in the graph, and then mark those ones that are actual matches for $P$. Unfortunately, the superposition would be as large as there are paths of length $|P|$, and thus the overall time complexity would be $O(|P|\sqrt{|V|^{|P|}})$. Moreover, an adjacency matrix would be needed to check the existence of edges between nodes in constant time, yielding a space complexity of $O(|V|^2)$ qubits. We conclude that more involved techniques are needed.

\subsection{The Classical Shift-And Algorithm}
We first introduce the classical \emph{shift-and} algorithm \cite{B-YG1992} for matching a pattern against a text and generalize it to work on graphs. Then, we show how the bit-vector data structure of that algorithm can be represented as a superposition of a logaritmic number of qubits. This approach allows us to achieve better performances than the brute force algorithm.

In the shift-and algorithm, we use bit vector $B$ of the same length of pattern $P$ to represent which of its prefixes are matching the text during the computation. Assuming integer-alphabet $\Sigma$, we also initialize bidimensional array $M$ of size $|P|\times|\Sigma|$ so that $M[j][\mathtt{c}]=1$ if and only if $P[j]=\mathtt{c}$, and $M[j][\mathtt{c}]=0$ otherwise. The algorithm starts by initializing vector $B$ to zero and array $M$ as specified above. Then, we scan whole text $T$ performing the next four operations for each $T[i]$, $i \in [0,n-1]$, where $M[*][\mathtt{c}]$ represents the $\mathtt{c}$-th column of $M$:
\begin{enumerate}
    \item \label{i1-c} $B \leftarrow B + 1$;
    \item \label{i2-c} $B \leftarrow B \land M[*][T[i]]$;
    \item \label{i3-c} if $B[m-1] = 1$, return \emph{yes};
    \item \label{i4-c} $B \leftarrow B<<1$.
\end{enumerate}
Operation \ref{i1-c} sets the least significant bit of $B$ to $1$, which is needed to test $P[0]$ against $T[i]$. Operation \ref{i2-c} computes a bit-wise \emph{and} between $B$ and the column of $M$ corresponding to character $T[i]$. Remember that $M[j][T[i]]=1$ means $P[j]=T[i]$, thus this operation leaves each bit $B[j]$ set to $1$ if and only if it was already set to $1$ before this step and  the the $j$-th character of the pattern matches the current character of the text. At this point, if bit $B[m-1]$ is set to 1 we have found a match for $P$, and Operation~\ref{i3-c} will return \emph{yes}. For the other positions, if bit $B[j]$ is set to $1$, then we know that prefix $P[0..j]$ matches $T[i-j+1..i]$, and Operation \ref{i4-c} shifts the bits in $B$ by one position, so that in the next iteration we will check whether $P[j+1]$ matches $T[i+1]$.

In labeled DAG $G=(V,E)$, each node $v_i \in V$ has a single-character label $\ell(v_i)$. We generalize the shift-and algorithm to labeled DAGs by computing a bit-vector $B_i$ for each node $v \in V$, initializing them to zero. Consider a BFS visit of DAG $G$. When visiting node $v_i$, each bit-vector $B_k$ of its in-neighbour $v_k \in in(v_i)$ represents a set of prefixes of $P$ matching a path in the graph ending at $v_k$. Thus, we merge all of this information together by taking the bit-wise \emph{or} of all of the in-neighbours of $v_i$, that is we replace Operation~\ref{i1-c} with $B_i \leftarrow 1+\bigvee_{v_k \in in(v_i)}B_k$. Operations~\ref{i2-c},~\ref{i3-c} and~\ref{i4-c} are performed as before. An example of the state of the data structures after the execution of the algorithm is shown in Figure~\ref{fig:shiftandex-LevelDAG}, and the body of the iteration now is:
\begin{enumerate}
    \item \label{i1-dag} $B_i \leftarrow 1+\bigvee_{v_k \in in(v)} B_k$;
    \item \label{i2-dag} $B_i \leftarrow B_i \land M[*][T[i]]$;
    \item \label{i3-dag} if $B_i[m-1] = 1$, return \emph{yes};
    \item \label{i4-dag} $B_i \leftarrow B_i<<1$.
\end{enumerate}

\subsection{Quantum Bit-Parallel Algorithm for Level DAGs}
We make the classic techniques work in a quantum setting for a special class of DAGs, which we call \emph{level} DAGs. A level DAG is a DAG such that, for every two nodes $v$ and $w$, every path from $v$ to $w$ has the same length, as for the DAG in Figure~\ref{fig:shiftandex-LevelDAG}. We also note that \emph{degenerate strings} \cite{Alzetal20} can be represented as level DAGs.
\begin{figure}
    \centering
    \includegraphics[scale=0.6]{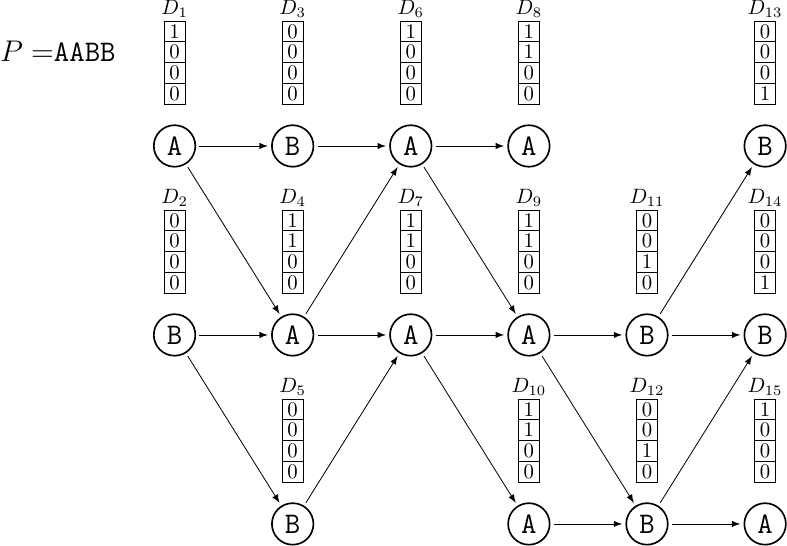}
    \caption{The adaptation of the classical algorithm for matching pattern $P$ in \emph{level} DAG $G$. Each bit-vector $D_v$ represent the result after the merging of the bit-vectors of the in-neighbours of $v$ and before the shifting.}
    \label{fig:shiftandex-LevelDAG}
\end{figure}
We use a function representing in-neighbours:
\[
    in_i(d) = \text{index of the } d \text{-th in-neighbour} \text{ of } v_i
\]

Our approach aims to represent each bit vector $B_i$ with a single qubit $V_i$ set up in a proper superposition, and translate the bit-wise operations to parallel operations across such superposition. In the algorithm, we use the following qubits and quantum registers. Quantum registers $I$ and $J$ store the index of a node and the position in the pattern, respectively. Qubit $V_i$ represents, in superposition, the bit-vector of the node $v_i$, and qubit $E_{i,d}$ stores the contribution of edge $(v_{in_i(d)},v_i) \in E$ in the update of qubit $V_i$, for, $0\leq i \leq n-1$, $0\leq d\leq D_i-1$ and $D_i=indeg(v_i)$. Quantum register $C$ stores label $\ell(v_i)$ of the node in the current iteration, and is used to fetch the content of the corresponding matrix column, which we will store in qubit $M$. Occurrences of the pattern encountered during the execution of the algorithm are stored in qubit $R_i$. Qubits $V_i'$ and $R_i'$ are auxiliary qubits used to store intermidiate results, and we also use auxiliary qubits $A$ and $B$ and auxiliary quantum register $Q$ to implement necessary operations. Moreover, we assume to have access to QRAM.

\subsubsection{The algorithm}
Assume all the quantum registers and qubits to be initialized to $\ket{0}$, except $Q$ initialized to $\ket{1}$. The algorithm starts by setting quantum register $J$ in a balanced superposition, by applying the Hadamard gate on each one of its qubits. Then, we initialize qubits $A$ so that $\ket{a_j}_A=\ket{1}$ for $j = 0$, and $\ket{a_j}_A=\ket{0}$ otherwise. We do the same with qubit $B$, with the difference that $\ket{b_j}_A=\ket{1}$ for $j = m-1$, and $\ket{b_j}_A=\ket{0}$ otherwise.
We can do these operations with two applications of a generalized Toffoli gate, using register $J$ as control and qubits $A$ and then $B$ as targets. In the case of qubit $A$, we first apply an $X$ gate to every qubit of register $J$, we then apply the Toffoli gate, and finally we undo the applications of the $X$ gate. The generalised Toffoli has a cost proportional to the number of qubits in $J$, that is logarithmic in the size of the input, and because this is an operation between a single quantum register and a qubit, we can assume it to be constant in the Word-QRAM model.
We then initialize the qubits representing the bit-vectors of the nodes at level $0$. This is done with the same operations described below for the main loop, the only difference being that these nodes do not have in-neighbours and thus we can simplify some operations. Specifically, we load each entry of the character matrix in superposition and we use it and qubit $A$ as controls of a Toffoli gate which thus flips to $\ket{1}$ sub-state $\ket{v_{i,j}}_{V_i}$ if $P[0]$ matches $\ell(v_i)$.

The rest of the algorithm maintains almost the same overall structure, with the exception of one necessary adaptation. In a DAG of $L$ levels where $L_l$ is the set of nodes at level $l$, for $0 \leq l\leq L-1$, we iterate over them one at the time, and for each level we process its nodes one after the other. As we will better explain later, we wait before applying the quantum equivalent of the shift operation once we scanned the whole level, not after processing every node. The overall idea is to translate the classical bit-parallel operations into analogous quantum operations that work across the superposition. This translation of bit-parallelism to superposition parallelism is the core of our technique, and we now describe how to apply it to each operation. The pseudocode of the entire procedure is given in Algorithm~\ref{alg:quantumSMLG}, where all the arithmetic operations are to be considere modulo $2^{|P|}$. We only omit the pseudocode for procedures \texttt{SourceNodesInit()}, \texttt{IncreaseI()} and \texttt{IncreaseJ()}, which is to be found in Appendix~\ref{appendix:pseudocode}. We also assume $|P|$ to be a power of two. If this is not the case, we generate a superposition as large as the first power of two greater than $|P|$, then standard techniques can be used to handle the additional substates, as explained in Appendix~\ref{appendix:power-of-two}. 

\textbf{Operation~\ref{i1-dag} (line \ref{algline:op1})} can be broken down into two simpler operations: computing the bit-wise $or$ and adding $1$. In our translation to quantum computing, each sub-state of superposition $\sum_{j=0}^{m-1}\ket{j}_J\ket{v_{i,j}}_{V_i}$ represents an entry of the classical bit-vector used in the Shift-And algorithm. Thus, what was a bit-wise $or$ is now easily translated into the application of few quantum gates. Notice that, to compute the logical $or$ between two generic qubits $P$ and $Q$ and store the result in qubit $R$, we can follow De Morgan's rules and apply an $X$ gate to both $P$ and $Q$, apply a Toffoli gate with controls $P$ and $Q$ and target $R$, apply an $X$ gate to $R$, and finally apply an $X$ gate to $P$ and $Q$ again to restore their initial values. In our case, at iteration $i$, we use qubit $E_{i,d}$ to store the $or$ computed among the first $d+1$ in-neighbours $v_{in_i(0)}, \ldots, v_{in_i(d)}$ of node $v_i$, and we compute it in the following way. Let
\[
\frac{1}{\sqrt{m}}\sum_{j=0}^{m-1}\ket{j}_J\ket{v_{in_i(0),j}}_{V_{in_i(0)}} \ket{v_{in_i(1),j}}_{V_{in_i(1)}} \cdots \ket{v_{in_i(d-1),j}}_{V_{in_i(d-1)}}\ket{e_{i,d-1,j}}_{E_{i,d-1}}
\]
be such that
\[
e_{i,d-1,j} = v_{in_i(0),j} \lor v_{in_i(1),j} \cdots \lor v_{in_i(d-1),j}.
\]
We compute the value of $E_{i,d}$ from $E_{i,d-1}$ and $V_{in_i(d)}$ as
\begin{align*}
    &\frac{1}{\sqrt{m}}\sum_{j=0}^{m-1}\ket{j}_J\ket{v_{in_i(d),j}}_{V_{in_i(d)}}\ket{e_{i,d-1,j}}_{E_{i,d-1}}\ket{0}_{E_{i,d}}\rightarrow\\
    &\frac{1}{\sqrt{m}}\sum_{j=0}^{m-1}\ket{j}_J\ket{v_{in_i(d),j}}_{V_{in_i(d)}}\ket{e_{i,d-1,j}}_{E_{i,d-1}}\ket{v_{in_i(d),j} \lor e_{i,d-1,j}}_{E_{i,d}}.
\end{align*}
Once we processed the last in-neighbour, $E_{i,D_i-1}$ stores the $or$ computed among all in-neighbours, where $D_i$ is the number of in-neighbours of node $v_i$.

We implement the classic operation of adding $1$ by computing an $or$ with qubit $A$ and storing the result in $V'_i$. Since $\ket{a_j}_A=\ket{\delta_{0,j}}$, we obtain $\ket{0}_{V'_i}\rightarrow\ket{v'_{i,j}}_{V'_i}$ where $\ket{v'_{i,j}}_{V'_i}=\ket{1}$ for $j=0$, while $\ket{v'_{i,j}}_{V'_i}=\ket{e_{i,D_i-1,j}}$ for $1\leq j \leq m-1$.

\textbf{Operation~\ref{i2-dag} (line \ref{algline:op2})} is implemented as a Toffoli-gate application with qubits $M$ and $V'_i$ as control and $V_i$ as target.
\[
\frac{1}{\sqrt{m}}\sum\limits_{j=0}^{m-1}\ket{m_{j,\ell(v_i)}}_M\ket{v'_{i,j}}_{V'_i}\ket{0}_{V_i}\rightarrow \frac{1}{\sqrt{m}}\sum\limits_{j=0}^{m-1}\ket{m_{j,\ell(v_i)}}_M\ket{v'_{i,j}}_{V'_i}\ket{m_{j,\ell(v_i)} \land v'_{i,j}}_{V_i}
\]

\textbf{Operation~\ref{i3-dag} (line \ref{algline:op3})}
is replaced by storing in register $R_i$ the presence of a match ending at node $v_i$. This requires an intermediate step in which we use qubit $B$ to filter the content of $V_i$. In fact, qubit $V_i$ now is in state $\ket{v_{i,j}}_{V_i}=\ket{1}$ for those values of $j$ such that $P[0..j]$ has a match ending at $v_i$ in the graph, and $\ket{v_{i,j}}_{V_i}=\ket{0}$ otherwise. Since we only care about potential full matches represented by $\ket{v_{i,m-1}}_{V_i}$, we use $B$, which is in state $\ket{\delta_{m-1,j}}_B$, as control qubit of a Toffoli gate, the other control qubit being $V_i$ and the target qubit being $R'_i$.
\[
    \frac{1}{\sqrt{m}}\sum\limits_{j=0}^{m-1}\ket{v_{i,j}}_{V_i}\ket{\delta_{m-1,j}}_{B}\ket{0}_{R'_i}\rightarrow
    \frac{1}{\sqrt{m}}\sum\limits_{j=0}^{m-1}\ket{v_{i,j}}_{V_i}\ket{\delta_{m-1,j}}_{B}\ket{v_{i,j} \land \delta_{m-1,j}}_{R'_i}
\]
Then, using the same technique as in Operation~\ref{i1-dag}, we compute an $or$ between $R'_i$ and $R_{i-1}$, storing the result in $R_i$.
\begin{align*}
    &\frac{1}{\sqrt{m}}\sum\limits_{j=0}^{m-1}\ket{v_{i,j} \land \delta_{m-1,j}}_{R'_i}\ket{r_{i-1,j}}_{R_{i-1}}\ket{0}_{R_i}\rightarrow\\
    &\frac{1}{\sqrt{m}}\sum\limits_{j=0}^{m-1}\ket{v_{i,j} \land \delta_{m-1,j}}_{R'_i}\ket{r_{i-1,j}}_{R_{i-1}}\ket{(v_{i,j} \land \delta_{m-1,j}) \lor r_{i-1,j}}_{R_i}
\end{align*}

After this operation, $\ket{r_{i,m-1}}_{R_i}$ is turned to $\ket{1}$ if there is a full match of $P$ ending at $v_i$, otherwise $\ket{r_{i,m-1}}_{R_i}$ is left unaltered.

\textbf{Operation~\ref{i4-dag} (line \ref{algline:op4})} consists in shifting all bits of the classical bit-vector by one position. In the quantum setting, we can perform this operation by adding $1$ to index register $J$ and then reorganising the sum: $\ket{1}_{C_\texttt{1}}\ket{j}_J \rightarrow \ket{1}_{C_\texttt{1}}\ket{j + 1}_J$. Notice that this changes value $\ket{j}_J$ in every term of the superposition to $\ket{j + 1}_j$. This can be interpreted as ``shifting'' value $k_j$ of generic register $K$ from $\ket{j}_J\ket{k_j}_K$ to $\ket{j + 1}_J\ket{k_j}_K$. Because this operation acts on every quantum register and qubit in this way, we have to reset qubits $A$ and $B$ to $\ket{0}$ before performing this operation and reinitialize their values afterwards, so that we prevent their values to be shifted. For the same reason, we also have to wait until having processed the whole level, otherwise we would shift the values of all the nodes at the previous level and compromise the computation.

As last step of the algorithm, we run Grover's search that uses as oracle function the whole procedure described up to this point, and then applies a $Z$ gate on qubit $R$. Thus, the marked sub-states are those such that $\ket{r_{i,j}}_{R_i} = \ket{1}$, which get mapped to $-\ket{r_{i,j}}_{R_i}$. Sub-states such that $\ket{r_{i,j}}_{R_i} = \ket{0}$ remain unaltered. As for the case of string matching in plain text, we rerun the whole algorithm a constant number of times to boost the probability of success, as explained in Theorem~\ref{theorem:QuantumSMLGCorrectness} and Appendix~\ref{appendix:theorem5-full-proof}. Algorithm~\ref{alg:quantumSMLG} shows the entire procedure.

\begin{algorithm}
\caption{Algorithm for testing whether pattern string $P$ has a match in level DAG $G$, running in $O(|E|\sqrt{|P|})$} time.
\label{alg:quantumSMLG}
    \SetKwFunction{FOperationOne}{OperationOne}
    \SetKwFunction{FOperationTwo}{OperationTwo}
    \SetKwFunction{FOperationThree}{OperationThree}
    \SetKwFunction{FOperationFour}{OperationFour}
    \SetKwFunction{FIncreaseI}{IncreaseI}
    \SetKwFunction{FIncreaseJ}{IncreaseJ}
    \SetKwFunction{FSourceNodesInit}{SourceNodesInit}
    \SetKwFunction{FDoubleSuperposition}{DoubleSuperposition}
    
    \KwIn{Graph $G$, pattern $P$, and constant $c$.}
    \KwOut{Returns $yes$ if $P$ occurs in $G$, otherwise $no$.}
    \For{$c$ times}{
    Initialize quantum register $Q$ to $\ket{1}$\;\label{algline:init}
    Initialize quantum registers $I,J,A,B,C,M,V_i,V'_i,R_i,R'_i,E_{i,d}\,$ to $\ket{0}$ where $i\in [0,n-1]$ and $d \in [0,D_i-1]$\;

    \tcp{Apply Hadamard to $J$}
    $\ket{0}_J \rightarrow \frac{1}{\sqrt{m}}\sum\limits_{j=0}^{m-1}\ket{j}_J$\;
    
    \FSourceNodesInit{$I,J,A,C,V_0,\ldots,V_{n-1},V'_0,\ldots,V'_{n-1}$}\;\label{algline:SourceNodesInit}\label{alg:start-oracle-dag}
    \FIncreaseJ{$J,A,B$}\;
    \tcp{$L$ is the number of levels}
    \For(\tcp*[f]{scan every level}){$l \in [1,L-1]$}{\label{algline:OuterForStart}
        
        \For(\tcp*[f]{scan every node in the level}){$|L_l|$\label{algline:MiddleForStart} times}{
            \FOperationOne{$l,I,C,M,E_{i,0},\ldots,E_{i,D_i-1},V'_i$}\;\label{algline:op1}
            \FOperationTwo{$M,V'_i,V_i$}\;\label{algline:op2}
            \tcp{Invariant 1 holds here}
            \FOperationThree{$B,V_i,R'_i,R_i$}\;\label{algline:op3}
            \FIncreaseI{$I,M,C$}\;\label{algline:IncI}
            
        }\label{algline:MiddleForEnd}
        \FOperationFour{$J,A,B$}\;\label{algline:op4}
        \tcp{Invariant 2 holds here}
    }\label{algline:OuterForEnd}
    
        Apply gate $Z$ to qubit $R_{n-1}$, so that the sign of the amplitude is flipped if $\ket{r_{n-1,j}}_{R_{n-1}}=\ket{1}$\label{alg:end-oracle-dag}\;
        
        Choose $K\in [0,|P|]$ uniformly at random\;

        Run Grover's iterate operator the optimal number of times assuming to have $K$ solutions, with the oracle function being lines~\ref{alg:start-oracle-dag}-\ref{alg:end-oracle-dag} of this algorithm\;
        
        Measure $R_{n-1}$ into classical register $R_{cl}$\;
        \If{$R_{cl}=1$}{\Return{$yes$}}
    }
    \Return{$no$}
\end{algorithm}

\begin{algorithm}
    \SetKwFunction{FOperationOne}{OperationOne}
    \SetKwFunction{FOperationTwo}{OperationTwo}
    \SetKwFunction{FOperationThree}{OperationThree}
    \SetKwFunction{FOperationFour}{OperationFour}
    
    \SetKwProg{Fn}{Function}{:}{}
    
    \Fn{\FOperationOne{$l,I,C,M,E_{i,0},\ldots,E_{i,D_i-1},V'_i$}}{
        \For(\tcp*[f]{scan every node in the level}){$|L_l|$ times}{
            \tcp{scan every node in $in(v_i)$}
            $\frac{1}{\sqrt{m}}\sum\limits_{j=0}^{m-1}\ket{j}_J\ket{i}_I\ket{0}_C \rightarrow \frac{1}{\sqrt{m}}\sum\limits_{j=0}^{m-1} \ket{j}_J\ket{i}_I\ket{\ell(v_i)}_C$\;
            
            $\frac{1}{\sqrt{m}}\sum\limits_{j=0}^{m-1}\ket{j}_J\ket{\ell(v_i)}_C\ket{0}_M \rightarrow \frac{1}{\sqrt{m}}\sum\limits_{j=0}^{m-1} \ket{j}_J\ket{\ell(v_i)}_C\ket{m_{\ell(v_i),j}}_M$\;
            
            $k \gets in_i(0)$\tcp*{Classical operation}
            $\frac{1}{\sqrt{m}}\sum\limits_{j=0}^{m-1}\ket{v_{k,j}}_{V_{k}}\ket{0}_{E_{i,0}}\rightarrow\frac{1}{\sqrt{m}}\sum\limits_{j=0}^{m-1}\ket{v_{k,j}}_{V_{k}}\ket{v_{k,j}}_{E_{i,0}}$\;
            
            \For(\tcp*[f]{scan every node in $in(v_i)$}){$d \in [1, D_i-1]$}{                
                $k \gets in_i(d)$\tcp*{Classical operation}
                \tcp{Add the contribution of the current in-neighbour}
                $\frac{1}{\sqrt{m}}\sum\limits_{j=0}^{m-1}\ket{v_{k,j}}_{V_{k}}\ket{e_{i,d-1,j}}_{E_{i,d-1}}\ket{0}_{E_{i,d}} \rightarrow \frac{1}{\sqrt{m}}\sum\limits_{j=0}^{m-1}\ket{v_{k,j}}_{V_{k}}\ket{e_{i,d-1,j}}_{E_{i,d-1}}\ket{e_{i,d-1,j}\lor v_{k,j}}_{E_{i,d}}$\;                
            }
            \tcp{Turn to $\ket{1}$ the substate corresponding to $j=0$}
            $\frac{1}{\sqrt{m}}\sum\limits_{j=0}^{m-1}\ket{\delta_{0,j}}_A\ket{e_{i,D_i-1,j}\lor v_{k,j}}_{E_{i,D_i-1}}\ket{0}_{V'_i} \rightarrow \frac{1}{\sqrt{m}}\sum\limits_{j=0}^{m-1}\ket{\delta_{0,j}}_A\ket{e_{i,D_i-1,j}\lor v_{k,j}}_{E_{i,D_i-1}}\ket{e_{i,D_i-1,j}\lor v_{k,j} \lor \delta_{0,j}}_{V'_i}$\;
        }
    }
    \bigskip
    \Fn{\FOperationTwo{$M,V'_i,V_i$}}{
        \tcp{Compute the $and$ with the column of the matrix}
        \mbox{$\frac{1}{\sqrt{m}}\sum\limits_{j=0}^{m-1}\ket{m_{\ell(v_i),j}}_M\ket{v'_{i,j}}_{V'_i}\ket{0}_{V_i}\rightarrow \frac{1}{\sqrt{m}}\sum\limits_{j=0}^{m-1}\ket{m_{\ell(v_i),j}}_M\ket{v'_{i,j}}_{V'_i}\ket{m_{\ell(v_i),j} \land v'_{i,j}}_{V_i}$\;}
    }
    \bigskip    
    \Fn{\FOperationThree{$B,V_i,R'_i,R_i$}}{
        \tcp{Set $\ket{r_{i,m-1}}_{R_i}=\ket{1}$ if there is a match ending at $v_i$}
        \tcp{Apply Toffoli on $V_i$, $B$ and $R'_i$}
        $\frac{1}{\sqrt{m}}\sum\limits_{j=0}^{m-1}\ket{v_{i,j}}_{V_i}\ket{\delta_{m-1,j}}_{B}\ket{0}_{R'_i} \rightarrow \frac{1}{\sqrt{m}}\sum\limits_{j=0}^{m-1}
        \ket{v_{i,j}}_{V_i}\ket{\delta_{m-1,j}}_{B}\ket{v_{i,j} \land \delta_{m-1,j}}_{R'_i}$\;
        
        \tcp{Apply logic $or$ on $R'_i$, $R_{i-1}$ and $R_i$}
        $\frac{1}{\sqrt{m}}\sum\limits_{j=0}^{m-1}\ket{v_{i,j} \land \delta_{m-1,j}}_{R'_i}\ket{r_{i-1,j}}_{R_{i-1}}\ket{0}_{R_i} \rightarrow \frac{1}{\sqrt{m}}\sum\limits_{j=0}^{m-1}\ket{v_{i,j} \land \delta_{m-1,j}}_{R'_i}\ket{r_{i-1,j}}_{R_{i-1}}\ket{(v_{i,j} \land \delta_{m-1,j}) \lor r_{i-1,j}}_{R_i}$\;
    }
    \bigskip    
    \Fn{\FOperationFour{$J,A,B$}}{
        \FIncreaseJ{$J,A,B$}\;
    }
\end{algorithm}

To prove the correctness of Algorithm~\ref{alg:quantumSMLG}, we formalise the key properties in the following lemmas. We start by ensuring that the shift operation provides the desired result. Let $l$ and $y$ be the total number of times that we started the execution of the middle \emph{for}-loop (lines~\ref{algline:MiddleForStart}-\ref{algline:MiddleForEnd}) and of the outer \emph{for}-loop (lines~\ref{algline:OuterForStart}-\ref{algline:OuterForEnd}), respectively. That is, $y=x+\sum_{\lambda=1}^{l-1}|L_\lambda|$ for $l\geq 2$, where $x \in [0,|L_l|]$ is the number of times that we started the execution of the middle \emph{for}-loop during the $l$-th iteration of the outer \emph{for}-loop. Notice that $y=0$ when $l=0$, and $y=x$ when $l=1$.

\begin{lemma}[\emph{Invariant 1}]
\label{lemma:forinvariant-Vi}
During the $l$-th execution of the outer \emph{for}-loop (lines~\ref{algline:OuterForStart}-\ref{algline:OuterForEnd}) and the $y$-th execution of the middle \emph{for}-loop (lines~\ref{algline:MiddleForStart}-\ref{algline:MiddleForEnd}) of Algorithm~\ref{alg:quantumSMLG}, but before the $y$-th execution of \texttt{OperationThree}() (line~\ref{algline:op4}), \emph{Invariant 1} holds: for every qubit $V_i$ such that $i \in L_l$ and $i \leq t$, we have substate $\ket{v_{i,j}}_{V_i}=\ket{1}$ if and only if there exists a path in $G$ ending at $v_i$ and matching $P[0,j]$, where $t=|L_0|+x-1$ is the index of the last node $v_t$ Algorithm~~\ref{alg:quantumSMLG} visited so far.
\end{lemma}
\begin{proof}
We proceed by strong induction on $y$, defined as above.

\textbf{Base case}, $y=0$. In this case, we executed the initialization but we have not run yet neither the outer nor the middle \emph{for}-loop. Thus, $l=0$, $t=|L_0|-1$, and qubits $V_i$ such that $i \in L_0$ and $i \leq t$ are those with in-degree zero, which are initialized by function \texttt{SourceNodesInit}(). For each such $i$, given that $J$ is in state $\sum_{j=0}^{m-1}\ket{j}_J$, function \texttt{SourceNodesInit}() first loads character $\ell(v_i)$ in register $C$ and matrix entry $m_{\ell(v_i),j}$ in register $M$, in superposition. Then, with regard to $t$, it performs transformation
\[
\sum\limits_{j=0}^{m-1}\ket{m_{\ell(v_i),j}}_{M}\ket{\delta_{0,j}}_{A}\ket{0}_{V_i}\rightarrow\sum\limits_{j=0}^{m-1}\ket{m_{\ell(v_i),j}}_{M}\ket{\delta_{0,j}}_{A}\ket{m_{\ell(v_i),j} \land \delta_{0,j}}_{V_i},
\]
where, by definition, $v_{t,j}=m_{\ell(v_i),j} \land \delta_{0,j}$. Thus, $\ket{v_{t,j}}=\ket{0}$ for every $j\neq 0$ because of $\delta_{0,j}$, and $\ket{v_{t,j}}_{V_i}=\ket{m_{\ell(v_i),j}}_{V_i}$ for $j=0$, which in turn means that $\ket{v_{t,0}}_{V_i}=\ket{1}$ if and only if $P[0,0]=\ell(v_i)$.

\textbf{Inductive case}, $y \geq 1$. We further divide our analysis in two sub-cases.

\textbf{First sub-case},  $x=|L_l|$. In this case, $y$ is the last iteration of the inner \emph{for}-loop during the $l$-th iteration of the outer \emph{for}-loop. We assume the inductive hypothesis to hold after the execution of \texttt{OperationTwo}(). We execute \texttt{OperationThree}() and \texttt{IncreaseI}(), which do not change the state of any $V_z$, for any $z \in [0,|V|-1]$. Now, we have to perform \texttt{OperationFour}() (line~\ref{algline:op4}) before starting iteration $y+1$ of the middle \emph{for}-loop, which will start iteration $l+1$ of the outer \emph{for}-loop. Assuming the inductive hypothesis, the application of \texttt{OperationFour}() makes every $V_z$ with $z\in L_l$ such that $\ket{j'}_J\ket{v_{z,j}}_{V_z}=\ket{j'}\ket{1}$, where $j'=j+1$, if and only if there is a match for $P[0,j]$ in $G$ ending at $v_z$, otherwise $\ket{j'}_J\ket{v_{z,j}}_{V_z}=\ket{j'}\ket{0}$. Then, we start iteration $y+1$ ($l+1$). Notice that we update $V_i$ if and only if $i\in L_{l+1}$ and, in any previous iteration of the middle \emph{for}-loop, this could have never been the case, thus every $\ket{v_{i,j}}_{V_i}$, $i \in L_{l+1}$, is currently set to $\ket{0}$. The same holds for every $V'_i$. The \emph{for}-loop inside \texttt{OperationOne}() computes a logic $or$ between all the qubits representing all the in-neighbours of $v_i$. Indeed, before running this \emph{for}-loop, we have $\ket{j'}_J\ket{v_{in_i(0),j}}_{E_{i,0}}$. After one iteration, we have $\ket{j'}\ket{v_{in_i(0),j} \lor v_{in_i(1),j}}_{E_{i,1}}$. After two iteration, we have $\ket{j'}\ket{v_{in_i(0),j} \lor v_{in_i(1),j} \lor v_{in_i(2),j}}_{E_{i,2}}$. After $D_i-1$ iterations, we have $\ket{j'}\ket{e_{i,D_i-1,j'}}_{E_{i,D_i-1}}$, where
\[
e_{i,D_i-1,j'} = \bigvee_{d=0}^{D_i-1} v_{in_i(d),j}.
\]
We store an intermediate result in $V'_i$, $\ket{v'_{i,j'}}_{V'_i}$, where $v'_{i,j'}=e_{i,D_i-1,j'}$ except for $j'=0$, because we make sure that $\ket{v'_{i,0}}_{V'_i}=\ket{1}$ thanks to the $or$ operation with qubit $A$, which stores $\ket{\delta_{0,j'}}_A$. Now we compute the logical $and$ with the entry of the matrix, as in the base case, obtaining $\ket{v_{i,j'}}_{V_i}$, where
\[
v_{i,j'} = m_{\ell(v_i),j'} \land (\delta_{0,j'} \lor \bigvee_{d=0}^{D_i-1} v_{in_i(d),j}).
\]
Applying the inductive hypothesis, this translates to
\begin{align*}
    v_{i,j+ 1} &= (P[j+1]=\ell(v_i)) \land \left((j+1 = 0) \lor \bigvee_{d=0}^{D_i-1} P[0..j]\text{ has a match ending at }v_{in_i(d)}\right)\\
    &= P[0..j+1]\text{ has a match ending at }v_i
\end{align*}
Thus, the statement of the lemma holds for $y+1$.

\textbf{Second sub-case}, $x<|L_l|$. The reasoning is analogous to the previous case, the only difference being that $j$ does not increase and thus we have to look back by $x+1$ iterations, when $j$ was increased the last time. This requires to assume that the inductive hypothesis was holding for iteration $y-x$, that is correct because, by strong induction, we assume the inductive hypothesis to hold for every $y' \leq y$ while proving the statement for $y+1$.
\end{proof}

\begin{lemma}[\emph{Invariant 2}]
\label{lemma:forinvariant-Ri}
After line~\ref{algline:OuterForEnd} of Algorithm~\ref{alg:quantumSMLG}, \emph{Invariant 2} holds: if there exists at least one match for $P$ in $G$ ending at some $v_i$ such that $i \leq t$, then there exists at least one $j$, $0 \leq j \leq m-1$, such that $\ket{r_{t,j}}_{R_i}=\ket{1}$, where $v_t$ is the last node we visited in Algorithm~\ref{alg:quantumSMLG} before line~\ref{algline:OuterForEnd}.
\end{lemma}
\begin{proof}
We proceed by induction on the number $l$ of times that we run the \emph{for}-loop at lines~\ref{algline:OuterForStart}-\ref{algline:OuterForEnd}.

\textbf{Base case}, $l=0$. In this case, nodes $v_t$ such that $t \in L_{l'}$, $l'\leq 0$ are those with in-degree zero, while the \emph{for}-loop at lines~\ref{algline:OuterForStart}-\ref{algline:OuterForEnd} has never run. Since we are visiting only single-node paths and we are assuming that pattern $P$ has length at least two, there can be no match for $P$ ending at these nodes. Correctly, $\ket{r_{i,j}}_{R_i}=\ket{0}$ for every $0\leq j\leq m-1$.

\textbf{Inductive case}, $1\leq l\leq L-1$.
By inductive hypothesis, we assume the statement of the lemma to be true right after running iteration $l$ of the \emph{for}-loop at lines~\ref{algline:OuterForStart}-\ref{algline:OuterForEnd}, and thus right before executing \texttt{IncreaseJ}() at line~\ref{algline:OuterForEnd}. After the execution of \texttt{IncreaseJ}(), the new state is $\sum_{j=0}^{m-1}\ket{j'}_J\ket{r_{i,j}}_{R_i}$, where $j'=j + 1$ and $v_i$ is the last node visited so far. Then, we start iteration $l+1$, processing $i'\in L_{l+1}$, $i'= t + 1$. We execute \texttt{OperationOne}() \texttt{OperationTwo}(), which do not affect register $R_{i'}$. Then we run the operations at lines~\ref{algline:op3}-\ref{algline:op3}, obtaining $\ket{r_{i',j'}}_{R_i}$ where $r_{i',j'} = (v_{i',j'}\land \delta_{m-1,j'})\lor r_{t,j}$. Let us consider the first time we run the middle \emph{for}-loop during iteration $l+1$ of the outer \emph{for}-loop. If $P$ has a match ending at some $v_z$, $z < i'$, the inductive hypothesis guarantees $r_{t,j}=1$ for some $j$. Otherwise, if $P$ does not have any such match, then $r_{t,j}=0$ for all $j$. In this second case, if $P$ has a match ending at $v_{i'}$, we know by Lemma~\ref{lemma:forinvariant-Vi} that $v_{i',m-1}=1$. This, combined with the fact that $\delta_{m-1,m-1}=1$, correctly implies that $r_{i',m-1}=1$, proving the statement for this specific $i'$ and $j'=m-1$. If $P$ has no match ending at $v_{i'}$, then $v_{i',m-1}=0$, and $r_{i',j'}=0$ for all $j'$, which must be the case when no match has been found yet. To conclude the proof, notice that the same reasoning applies for the subsequent iterations of the middle \emph{for}-loop by using every time the previous instance of this reasoning in place of the inductive hypothesis. That is, we use $r_{i',j'}$ when proving the statement for $r_{i'+ 1,j'}$ and so on, until we prove the statement for $r_{t',j'}$, where $v_{t'}$ is the last node with index in $L_{l+1}$. At this point, we exit the middle \emph{for}-loop and the statement of the lemma is proven for $l+1$.
\end{proof}

The correctness of the algorithm follows from the previous lemma combined with few additional observations.
\begin{theorem}
\label{theorem:QuantumSMLGCorrectness}
Given pattern string $P$ of length at least $2$ and level DAG $G$,  Algorithm~\ref{alg:quantumSMLG} returns the right answer for the \smlg problem on $P$ and $G$ with probability $p>1-(7/8)^c$, for any given integer $c$.
\end{theorem}
\begin{proof}
After running the outer \emph{for}-loop of Algorithm~\ref{alg:quantumSMLG} $L-1$ times, we exit such a loop, and we know we have visited all the nodes (nodes in $L_0$ where visited during the initialization). If we consider Lemma~\ref{lemma:forinvariant-Ri} applied in the case of $t=n-1$, we are considering all the nodes, which means that if $P$ has no match ending in $G$, then no substate of register $R_{n-1}$ is such that $\ket{r_{n-1,j}}_{R_{n-1}}=\ket{1}$, for any $j$. Instead, if $P$ has a match in $G$, then at least one substate of $R_{n-1}$ is such that $\ket{r_{n-1,j}}_{R_{n-1}}=\ket{1}$, for some $j$. We use standard techniques that consist in rerunning the algorithm a constant number of times to boost the probability of measuring such a state, and achieve the desired one. Appendix~\ref{appendix:theorem5-full-proof} provides a more detailed analysis.
\end{proof}

Finally, the time complexity of our algorithm is subquadratic in the size of the graph.
\begin{theorem}
The time complexity of Algorithm~\ref{alg:quantumSMLG} is $O(|E|\sqrt{|P|})$ in the QRAM model, and the space complexity is $O(|E|+|V|)$.
\end{theorem}
\begin{proof}
The algorithm uses $|V|$ qubits $V_i$, and the same amount of qubits $V'_i$, $R_i$, $R'_i$; qubits $E_{i,d}$ are a total of $|E|$ qubits, and the rest are a constant number of qubits and registers. Thus, the space complexity is $O(|E|+|V|)$.

With the \emph{for}-loop in function \texttt{SourceNodesInit}(), the algorithm visits the nodes in $L_0$, which are at most $O(|V|)$. The iteration conditions at lines~\ref{algline:OuterForStart}~and~\ref{algline:MiddleForStart} make the algorithm visit every node. For each such iteration, we perform a constant number of operations except for the \emph{for}-loop in \texttt{OperationOne}(). This \emph{for}-loop visits all the in-neighbours of a node, each time performing a constant number of operations, and $\sum_{i=0}^{|V|-1}|in(i)|=|E|$. All of the aforementioned operations can be implemented with a constant number of quantum-gate applications, each affecting a constant number of qubits ($3$ at most), or by performing a load operation from the QRAM, assumed to require constant time. At the end of the algorithm, we run Grover's search procedure on a superposition of $2|P|$ states, using the entire algorithm as the oracle function.

Summing everything up, we spend $O(|V|)$ time for the initialization, $O(|V|+|E|)$ time in the \emph{for}-loops, and $O(|E|\sqrt{|P|})$ time for Grover's search procedure. The total time complexity is thus dominated by $O(|E|\sqrt{|P|})$. 
\end{proof}

\bibliography{biblio}

\appendix

\section{Reductions to the power-of-two case}
\label{appendix:power-of-two}
For string matching in plain text, if $|T|$ is not a power of two, in addition to quantum register $I$ for indexing, we use also quantum register $I'$, of the same size. Let $x$ be the only integer $x$ such that $|T|<2^x<2|T|$. Generate superposition $\sum\limits_{i'=0}^{2^x-1}\ket{i'}_{I'}\ket{0}_I$ and compute $\sum\limits_{i'=0}^{2^x-1}\ket{i'}_{I'}\ket{i' \,mod\, |T|}_I$. Now run the algorithm using $I$ as normal. This creates some redundant substates, but does not affect the correctness of the algorithm.

For \smlg in level DAGs, if $|P|$ is not a power of two, we generate a superposition of size $2^x$, where $x$ is the only integer such that $|P|<2^x<2|P|$. Then, it sufficies to assume that every entry that we read from the QRAM to the additional substates between $|P|$ and $2^x$ is always initialized to $\ket{1}$, because this is the neutral value in a logical \emph{and}, an thus in the application of the Toffoli gate. Therefore, in these substates, a qubit $R_i$ can and will be set to value $\ket{1}_{R_i}$ if and only if a previous ``shift'' carried $\ket{1}_{R_{i-1}}$.

Alternatively, if $|P|$ is not a power of two, we can classically reduce the problem to this case. We add new symbol $\$$ to the alphabet. Then, we pad $P$ with as many $\$$ at the end as needed to reach the next power of two. For each level in the DAG, we add a new node with label $\$$, and we place an edge for every node in that level to the new node. We connect all this new nodes in a chain, and we also add a chain of $|P|$ such nodes after the last level (they create new levels consisting only of one node). The pattern now can overflow in these nodes after finding a proper match in the DAG. Finally, we apply the same binary encoding as in the plain text case, now replacing every node with a chain of two nodes, sending all the incoming edges to the first node and making all the outgoing edges leave from the second node. Overall, we add one new node per level, and one new edge per node, plus $|P|$ additional nodes and edges after the last level. This takes time $O(|E|+|P|)$.

\section{Additional pseudo-code}
\label{appendix:pseudocode}

\begin{algorithm}
    \SetKwFunction{FSourceNodesInit}{SourceNodesInit}

    \SetKwProg{Fn}{Function}{:}{}
    
    \Fn{\FSourceNodesInit{$I,J,A,C,V_0,\ldots,V_{n-1},V'_0,\ldots,V'_{n-1}$}}{        
        \tcp{Initialize $\ket{a_j}_A$ so that $a_j=1$ if $j=0$, $a_j=0$ otherwise}
        $\frac{1}{\sqrt{m}}\sum\limits_{j=0}^{m-1}\ket{j}_J\ket{0}_A \rightarrow \frac{1}{\sqrt{m}}\sum\limits_{j=0}^{m-1}\ket{j}_J\ket{\delta_{0,j}}_A $\;
        
        \tcp{Initialize $\ket{b_j}_B$ so that $b_j=1$ if $j=m-1$, $b_j=0$ otherwise}
        $\frac{1}{\sqrt{m}}\sum\limits_{j=0}^{m-1}\ket{j}_J\ket{0}_B \rightarrow \frac{1}{\sqrt{m}}\sum\limits_{j=0}^{m-1}\ket{j}_J\ket{\delta_{m-1,j}}_B $\;    
        
        \tcp{$L_0$ is the set of nodes in level $0$.}
        \For{$|L_0|$ times}{
            \tcp{Read node label $\ell(v_i)$ in $C$}
            $\frac{1}{\sqrt{m}}\sum\limits_{j=0}^{m-1}\ket{j}_J\ket{i}_I\ket{0}_{C} \rightarrow \frac{1}{\sqrt{m}}\sum\limits_{j=0}^{m-1}\ket{j}_J\ket{i}_I\ket{\ell(v_i)}_{C}$\;
            
            \tcp{Read the matrix entries for character $\ell(v_i)$ in $M$}
            $\frac{1}{\sqrt{m}}\sum\limits_{j=0}^{m-1}\ket{j}_J\ket{\ell(v_i)}_{C}\ket{0}_{M} \rightarrow \frac{1}{\sqrt{m}}\sum\limits_{j=0}^{m-1}\ket{j}_J\ket{\ell(v_i)}_{C}\ket{m_{\ell(v_i),j}}_{M}$\;
            
            \tcp{Apply Toffoli to qubits $M$, $A$ and $V_i$}
            $\frac{1}{\sqrt{m}}\sum\limits_{j=0}^{m-1}\ket{m_{\ell(v_i),j}}_{M}\ket{\delta_{0,j}}_{A}\ket{0}_{V_i}\rightarrow
            \frac{1}{\sqrt{m}}\sum\limits_{j=0}^{m-1}\ket{m_{\ell(v_i),j}}_{M}\ket{\delta_{0,j}}_{A}\ket{m_{\ell(v_i),j} \land \delta_{0,j}}_{V_i}$\;
            
            \tcp{Reset $M$ and $C$}
            $\frac{1}{\sqrt{m}}\sum\limits_{j=0}^{m-1}\ket{j}_J\ket{\ell(v_i)}_C\ket{m_{\ell(v_i),j}}_M \rightarrow \frac{1}{\sqrt{m}}\sum\limits_{j=0}^{m-1}\ket{j}_J\ket{\ell(v_i)}_C\ket{m_{\ell(v_i),j} \oplus m_{\ell(v_i),j}}_M = \frac{1}{\sqrt{m}}\sum\limits_{j=0}^{m-1}\ket{j}_J\ket{\ell(v_i)}_{C_i}\ket{0}_M$\;
            
            $\frac{1}{\sqrt{m}}\sum\limits_{j=0}^{m-1}\ket{j}_J\ket{i}_I\ket{\ell(v_i)}_{C} \rightarrow$ \mbox{$\frac{1}{\sqrt{m}}\sum\limits_{j=0}^{m-1}\ket{j}_J\ket{i}_I\ket{\ell(v_i) \oplus \ell(v_i)}_{C} = \frac{1}{\sqrt{m}}\sum\limits_{j=0}^{m-1}\ket{j}_J\ket{i}_I\ket{0}_{C}$}\;
            
            \tcp{Increase $I$ by one to visit the next node}
            $\frac{1}{\sqrt{m}}\sum\limits_{j=0}^{m-1}\ket{1}_Q\ket{i}_I\rightarrow \frac{1}{\sqrt{m}}\sum\limits_{j=0}^{m-1}\ket{1}_Q\ket{i + 1}_I$\;
        }
    }
\end{algorithm}

\begin{algorithm}
    \SetKwFunction{FOperationFour}{OperationFour}
    \SetKwFunction{FIncreaseI}{IncreaseI}
    \SetKwFunction{FIncreaseJ}{IncreaseJ}
    \SetKwFunction{FDoubleSuperposition}{DoubleSuperposition}

    \SetKwProg{Fn}{Function}{:}{}
    
    \Fn{\FIncreaseI{$I,M,C$}}{
        \tcp{Reset $M$ and $C$}
        $\frac{1}{\sqrt{m}}\sum\limits_{j=0}^{m-1}\ket{j}_J\ket{\ell(v_i)}_C\ket{m_{\ell(v_i),j}}_M \rightarrow \frac{1}{\sqrt{m}}\sum\limits_{j=0}^{m-1}\ket{j}_J\ket{\ell(v_i)}_C\ket{m_{\ell(v_i),j} \oplus m_{\ell(v_i),j}}_M \rightarrow \frac{1}{\sqrt{m}}\sum\limits_{j=0}^{m-1}\ket{j}_J\ket{\ell(v_i)}_{C_i}\ket{0}_M$\;
        
        $\frac{1}{\sqrt{m}}\sum\limits_{j=0}^{m-1}\ket{j}_J\ket{i}_I\ket{\ell(v_i)}_C \rightarrow \frac{1}{\sqrt{m}}\sum\limits_{j=0}^{m-1}\ket{j}_J\ket{i}_I\ket{\ell(v_i) \oplus \ell(v_i)}_C = \frac{1}{\sqrt{m}}\sum\limits_{j=0}^{m-1}\ket{j}_J\ket{i}_I\ket{0}_C$\;
        
        \tcp{Increase $I$}
        $\frac{1}{\sqrt{m}}\sum\limits_{j=0}^{m-1}\ket{j}_J\ket{1}_Q\ket{i}_I\rightarrow\frac{1}{\sqrt{m}}\sum\limits_{j=0}^{m-1}\ket{j}_J\ket{1}_Q\ket{i + 1}_{I}$\;
    }
\end{algorithm}

\begin{algorithm}
    \SetKwFunction{FOperationFour}{OperationFour}
    \SetKwFunction{FIncreaseI}{IncreaseI}
    \SetKwFunction{FIncreaseJ}{IncreaseJ}
    \SetKwFunction{FDoubleSuperposition}{DoubleSuperposition}

    \SetKwProg{Fn}{Function}{:}{}
    \Fn{\FIncreaseJ{$J,A,B$}}{
        \tcp{Reset $A$ and $B$}
        $\frac{1}{\sqrt{m}}\sum\limits_{j=0}^{m-1}\ket{j}_J\ket{\delta_{0,j}}_A \rightarrow  \frac{1}{\sqrt{m}}\sum\limits_{j=0}^{m-1}\ket{j}_J\ket{0}_A$\;
        $\frac{1}{\sqrt{m}}\sum\limits_{j=0}^{m-1}\ket{j}_J\ket{\delta_{m-1,j}}_B \rightarrow  \frac{1}{\sqrt{m}}\sum\limits_{j=0}^{m-1}\ket{j}_J\ket{0}_B$\;
        
        \tcp{Increase $J$}
        $\frac{1}{\sqrt{m}}\sum\limits_{j=0}^{m-1}\ket{1}_Q\ket{j}_J\rightarrow\frac{1}{\sqrt{m}}\sum\limits_{j=0}^{m-1}\ket{1}_Q\ket{j + 1}_{J}$\;
        
        \tcp{Reinitialize $A$ and $B$}
        $\frac{1}{\sqrt{m}}\sum\limits_{j=0}^{m-1}\ket{j}_J\ket{0}_A \rightarrow \frac{1}{\sqrt{m}}\sum\limits_{j=0}^{m-1}\ket{j}_J\ket{\delta_{0,j}}_A$\;
        $\frac{1}{\sqrt{m}}\sum\limits_{j=0}^{m-1}\ket{j}_J\ket{0}_B \rightarrow \frac{1}{\sqrt{m}}\sum\limits_{j=0}^{m-1}\ket{j}_J\ket{\delta_{m-1,j}}_B $\;
    }
\end{algorithm}

\section{Full proof of Theorem~\ref{theorem:QuantumSMLGCorrectness}.}
\label{appendix:theorem5-full-proof}

\begin{proof}
After running the outer \emph{for}-loop of Algorithm~\ref{alg:quantumSMLG} $L-1$ times, we exit such a loop, and we know we have visited all the nodes (nodes in $L_0$ where visited during the initialization). If we consider Lemma~\ref{lemma:forinvariant-Ri} applied in the case of $t=n-1$, we are considering all the nodes, which means that if $P$ has no match ending in $G$, then no substate of register $R_{n-1}$ is such that $\ket{r_{n-1,j}}_{R_{n-1}}=\ket{1}$, for any $j$. Instead, if $P$ has a match in $G$, then at least one substate of $R_{n-1}$ is such that $\ket{r_{n-1,j}}_{R_{n-1}}=\ket{1}$, for some $j$.

The \emph{for} loop that we run at the end of the algorithm ensures to achieve high probability of success. The probability of success $p$ in Grover's search algorithm is the sinusoidal function $p(K)=\sin^2((2K+1)\theta)$ \cite{BBGHT98}, where $\theta=\sin^{-1}\left(\sqrt{\frac{M}{N}}\right)$, $N$ is the search space, $M$ is the number of good solutions and $K$ is the number of iterations of the Grover's operator. This function has period $\lambda_M \approx \frac{\pi}{2}\sqrt{\frac{N}{M}}-1$. Consider the case $M=1$. If we choose a random number of iterations $K$ between $1$ and $\lambda_1$, we have $p(K)\geq 1/2$ with probability $1/2$. This is because half of the material of the function is above the horizontal line of $1/2$. When $p(K)\geq1/2$, the probability of measuring a wrong result is $p_{\text{top}}\leq 1/2$. When $p(K)\geq1/2$, the probability of measuring a wrong result is greater than $1/2$, but anyway $p_{\text{bottom}}\leq 1$. If we run the process $c$ times, the overall probability of failure (measuring a wrong result) $p_f$ is then
\[
p_f = \left(p_{\text{top}}\frac12 + p_{\text{bottom}}\frac12\right)^c \leq \left(\frac12\frac12 + 1\cdot \frac12\right)^c = \left(\frac34\right)^c
\]
Thus, the probability of success (measuring a correct result) is $p_s=1-(3/4)^c$.

In the general case $1<M\leq N$, the period $\lambda_M$ of function $p(K)$ is smaller than period $\lambda_1$ of the case $M=1$. We can still use the same random number of iterations $K$ between $1$ and $\lambda_1$, as nearly half of the material of the function $p(K)$ is above the horizontal line of $1/2$: the worst case is when $\lambda_M$ is little over half of $\lambda_1$. In this case we know that $p(K)$ will be sampled uniformly over half of the range of period $\lambda_1$, but the other half may have biased sampling. Namely, the other half of the function might have more material below $1/2$ than above. To have a safe estimate, we assume that the probability of returning the wrong result in the biased case is $p_{\text{biased}}=1$. That is, if we run the process $c$ times, the overall probability of failure (measuring a wrong result) $p_f$ is then
\[
p_f = \left(p_{\text{biased}}\frac12 + (p_{\text{top}}\frac12 + p_{\text{bottom}}\frac12)\frac12\right)^c \leq \left(1\cdot \frac12+(\frac12\frac12 + 1\cdot \frac12)\frac12\right)^c = \left(\frac78\right)^c
\]
Thus, the probability of success (measuring a correct result) is $p_s=1-(7/8)^c$.  

\end{proof}

\end{document}